\documentclass[12pt,twoside]{article}

 \usepackage{float}
\usepackage{graphicx}
\usepackage{epstopdf}
\usepackage{graphicx}
\usepackage{epic}
\usepackage{multirow}
\usepackage{pst-poly}  
\usepackage{pst-plot}  
\usepackage{pst-poly}  
\usepackage{tikz}
\usepackage{xcolor}
\usetikzlibrary{arrows,shapes,chains}

\renewcommand{\paragraph}{\roman{paragraph}}
\usepackage[a4paper]{geometry}
\makeatletter
\renewcommand\title[1]{\gdef\@title{\reset@font\Large\bfseries #1}}
\renewcommand\section{\@startsection {section}{1}{\z@}%
                                   {-3.5ex \@plus -1ex \@minus -.2ex}%
                                   {2.3ex \@plus.2ex}%
                                   {\normalfont\large\bfseries}}
\renewcommand\subsection{\@startsection{subsection}{2}{\z@}%
                                     {-3ex\@plus -1ex \@minus -.2ex}%
                                     {1.5ex \@plus .2ex}%
                                     {\normalfont\normalsize\bfseries}}
\renewcommand\subsubsection{\@startsection{subsubsection}{3}{\z@}%
                                     {-2.5ex\@plus -1ex \@minus -.2ex}%
                                     {1.5ex \@plus .2ex}%
                                     {\normalfont\normalsize\bfseries}}

\def\@runningauthor{}\newcommand{\runningauthor}[1]{\def\runningauthor{#1}}
\def\@runningtitle{}\newcommand{\runningtitle}[1]{\def\runningtitle{#1}}

\renewcommand{\ps@plain}{%
\renewcommand{\@evenhead}{\footnotesize\scshape \hfill\runningauthor\hfill}
\renewcommand{\@oddhead}{\footnotesize\scshape \hfill\runningtitle\hfill}}
\newcommand{\Z}{\mathbb{Z}}
\newcommand{\F}{\mathbb{F}}

\g@addto@macro\bfseries{\boldmath}

\makeatother

\usepackage{amsthm,amsmath,amssymb}
\usepackage{cite}
\usepackage{graphicx}

\usepackage[colorlinks=true,citecolor=black,linkcolor=black,urlcolor=blue]{hyperref}

\theoremstyle{plain}
\newtheorem{theorem}{Theorem}[section]

\newtheorem{lemma}[theorem]{Lemma}
\newtheorem{cor}[theorem]{Corollary}
\newtheorem{prop}[theorem]{Proposition}

\theoremstyle{definition}
\newtheorem{definition}[theorem]{Definition}
\newtheorem{example}[theorem]{Example}

\theoremstyle{remark}
\newtheorem{remark}[theorem]{Remark}

\runningauthor{}

\date{}

\begin{document}

 \title{An improved method for constructing linear codes with small hulls\thanks{This research is supported by the National Natural Science Foundation of China (12071001) and
the Excellent Youth Foundation of Natural Science Foundation of Anhui Province (1808085J20).}}
\author{ Shitao Li\thanks{lishitao0216@163.com}, Minjia Shi\thanks{smjwcl.good@163.com}
\thanks{ Shitao Li and Minjia Shi are with School of Mathematical Sciences, Anhui University, Hefei, China.}}

\date{}
    \maketitle

\begin{abstract}
In this paper, we give a method for constructing linear codes with small hulls by generalizing the method in \cite{LCD-T-matric}.
As a result, we obtain many optimal Euclidean LCD codes and Hermitian LCD codes, which improve the previously known lower bound on the largest minimum distance. We also obtain many optimal codes with one-dimension hull. Furthermore, we give three tables about formally self-dual LCD codes.

\end{abstract}
{\bf AMS Classification (MSC 2020)}: 94B05, 15B05, 12E10

\section{Introduction}
Let $\F_q$ denote the finite field with $q$ elements, where $q$ is a prime power.
The Euclidean hull of a linear code $C$ over $\F_q$ is the intersection of the code and its dual code with respect to the Euclidean inner product.
In particular, if $C\cap C^{\perp_E}=\{0\}$, the code $C$ is called Euclidean linear complementary dual (Euclidean LCD).
The Hermitian hull of a linear code over $\F_{q^2}$ is the intersection of the code and its dual code with respect to the Hermitian inner product.
In particular, if $C\cap C^{\perp_H}=\{0\}$, the code $C$ is called Hermitian linear complementary dual (Hermitian LCD).

The Euclidean hull was introduced in 1990 by Assmus and Key \cite{A-hull-DM} to classify finite projective planes.
It had showed that the Euclidean hull plays an important role in determining the complexity of the algorithms for checking permutation equivalence of two linear codes and for computing the automorphism group of a linear code \cite{J-p-group,Sen-p-e-codes,Sen-S-auto-group}. When the size of the hull is small, these algorithms are very effective in general.
Moreover, Sendrier \cite{LCD-is-good} showed that LCD codes meet the asymptotic Gilbert-Varshamov bound by using the hull dimension spectra of linear codes in 2004.
In 2014, Carlet et al. \cite{lcd-appl} investigated an application of binary LCD codes against Side-Channel Attacks (SCA) and Fault Injection Attack (FIA), and gave several constructions of LCD codes.
Recently, LCD codes and linear codes with one-dimension hull were extensively studied \cite{Li,zhu,C-G-O-S-LCD-and-isometry, one-C-L-M,LCD-5,Li-one-hull,one-qian-ccds,one-qian,LCD-qian,LCD-9,LCD-huang}.
However, less results have been known for codes with Hermitian hulls. G$\ddot{{\rm u}}$neri et al. proved that Hermitian LCD codes are asymptotically good in \cite{LCD-6}.
Boonniyoma and Jitman \cite{Hermitian-LCD} gave a necessary
and sufficient condition for Hermitian codes to be LCD.
The authors of \cite{H-lcd-1} and \cite{H-LCD} gave the constructions of Hermitian LCD codes.

On the other hand, it is a fundamental topic to determine the largest minimum distance of LCD codes for various lengths and dimensions. Recently, much
work has been done concerning this topic (see \cite{bound-13,HS-BLCD-1-16,AH-BLCD-17-24,B-LCD-40,AH-TLCD,H-LCD,T-11-19,2-LCD-30}).
Carlet et al. \cite{LCD-equivalent} showed that any code over $\F_q$ is equivalent to some Euclidean LCD code for $q\geq 4$ and any code over $\F_{q^2}$ is equivalent to some Hermitian LCD code for $q\geq 3$. This motivates us to study binary LCD codes, ternary LCD codes, and quaternary Hermitian LCD codes.
Based on the above discussions, it is desired to construct linear codes with a small hull.

A matrix is called Toeplitz if it has constant entries on all diagonals parallel to the main diagonal.
A double Toeplitz code is a linear code with the generator matrix $(I\ A)$, where $I$ is the identity matrix, $A$ is a Toeplitz matrix of the same order. These codes have been proved to achieve a modified Gilbert-Varshamov bound \cite{double-T-codes}. Very recently, Shi et al. \cite{LCD-T-matric} first characterized a class of double Toeplitz codes as LCD codes by using the factorization of Dickson polynomials \cite{Factoring-DK-Po-}.

In this paper, inspired by the method in \cite{LCD-T-matric}, we further generalize this method.
As a result, many new Euclidean LCD
codes and Hermitian LCD codes are obtained, which improve the previously known lower bound on the largest minimum distance compared with the latest code tables \cite{B-LCD-40,H-LCD}.
We also obtain some optimal codes with one-dimension hull.
More importantly, we give three tables about formally self-dual LCD codes. This seems to be the first article to give formally self-dual LCD codetables.

The paper is organized as follows. In Section 2, we give some notations and definitions. In Section 3, we give some preliminary results. In Section 4, we study constructions of linear codes small Euclidean hulls based on symmetric tridiagonal Toeplitz matrices. In Section 5, we study constructions of linear codes with small Hermitian hulls based on Hermitian tridiagonal Toeplitz matrices. In Section 6, we characterize a class of codes as LCD codes. In Section 7, we conclude the paper.

\section{Preliminaries}
\subsection{Codes}
For any $\textbf x\in \F_q^N$, the Hamming weight of $\textbf x$ is the
number of nonzero components of $\textbf x$.
An $[N,K,D]$ linear code $C$ over $\F_q$ is a $K$-dimension subspace of $\F_q^N$, where $D$ is the minimum nonzero Hamming weight of $C$.
The Euclidean dual code $C^{\perp_E}$ of a linear code $C$ is defined as
$$C^{\perp_E}=\{\textbf y\in \F_q^N~|~\langle \textbf x, \textbf y\rangle_E=0, {\rm for\ all}\ \textbf x\in C \},$$
where $\langle \textbf x, \textbf y\rangle_E=\sum_{i=1}^N x_iy_i$ for $\textbf x = (x_1,x_2, \ldots, x_N)$ and $\textbf y = (y_1,y_2, \ldots, y_N)\in \F_q^N$.
The Hermitian dual code $C^{\perp_H}$ of a linear code $C$ over $\F_{q^2}$ is defined as
$$C^{\perp_H}=\{\textbf y\in \F_{q^2}^N~|~\langle \textbf x, \textbf y\rangle_H=0, {\rm for\ all}\ \textbf x\in C \},$$
where $\langle \textbf x, \textbf y\rangle_H=\sum_{i=1}^N x_i\overline{y_i}$ for $\textbf x = (x_1,x_2, \ldots, x_N)$ and $\textbf y = (y_1,y_2, \ldots, y_N)\in \F_{q^2}^N$. Note that $\overline{x}=x^q$ for any $x\in \F_{q^2}$.
The Euclidean hull (resp. Hermitian hull) of the linear code $C$ over $\F_q$ (resp. $\F_{q^2}$) is defined as
$${\rm Hull_E}(C)=C\cap C^{\perp_E}\ ({\rm resp.\ Hull_H}(C)=C\cap C^{\perp_H}).$$

It is easy to see that ${\rm Hull}_E(C)$ (resp.\ ${\rm Hull}_H(C)$) is a linear code over $\F_q$ (resp. $\F_{q^2}$).
Suppose that the dimension of ${\rm Hull}_E(C)$ (resp. ${\rm Hull}_H(C)$) is $l$.
If $l=0$, that is to say $C\cap C^{\perp_E}=\{\textbf 0\}$ (resp. $C\cap C^{\perp_H}=\{\textbf 0\}$), the code $C$ is called a Euclidean linear complementary dual (Euclidean LCD) (resp. Hermitian linear complementary dual (Hermitian LCD)) code.
If $l=K$, that is to say $C\subseteq C^{\perp_E}$ (resp. $C\subseteq C^{\perp_H}$), the code $C$ is called a Euclidean self-orthogonal (resp. Hermitian self-orthogonal) code. In addition, $l=\frac{N}{2}$ for even $N$, that is to say $C= C^{\perp_E}$ (resp. $C= C^{\perp_H}$), the code $C$ is Euclidean self-dual (resp. Hermitian self-dual).

The weight distribution of a code $C$ is the sequence
of integers $A_i'$s for $i = 0, 1, \ldots, n$, where $A_i$ is the number of codewords of
weight $i$. A code is Euclidean (resp. Hermitian) formally self-dual (FSD) if it has the same weight distribution as its dual with respect to the Euclidean (resp. Hermitian) inner product.

A linear $[N,K,D]$ code $C$ is {\bf optimal} if $C$ has the largest minimum distance among all linear $[N,K]$ codes, and $C$ is called {\bf almost optimal} if there exists an optimal $[N,K,D+1]$ code.
An LCD $[N, K]$ code with the largest minimum distance among all LCD $[N, K]$ codes is {\bf optimal LCD}. And $C$ is called {\bf almost optimal LCD} if there exists an optimal $[N,K,D+1]$ LCD code.
It is well-known that the Griesmer bound \cite[Chap. 17, Section 5]{MacWilliams} on a linear $[N, K, D]$ code over $\F_q$ is given by
$N\geq \sum_{i=0}^{K-1}\left\lceil \frac{D}{q^i}\right\rceil$,
where $\lceil \cdot\rceil$ is the least integer greater than or equal to $\cdot$.

For a matrix $A = (a_{ij} )$, let $A^T$ denote the transpose of $A$.
The conjugate matrix of $A$ is defined as $\overline{A} = (\overline{a_{i j}} )$.
Let $\{{\bf n_1}\cdot a_1,{\bf n_2} \cdot a_2,\ldots,{\bf n_k}\cdot a_k\}$ denote a multiset, where ${\bf n_i}$ represents the multiplicity of $a_i$ $(1\leq 1\leq k)$.

\subsection{The eigenvalues of $T_n(a,b,c)$ and $T'_n(a,b,c)$}
Let $T_n(a,b,c)$ be an $n\times n$ tridiagonal Toeplitz matrix over $\F_q$, defined as
$$T_n(a,b,c)=\begin{pmatrix}
    \begin{array}{cccccc}
    a & c & 0 & \cdots & 0 & 0 \\
    b & a & c & \cdots & 0 & 0 \\
    0 & b & a & \cdots & 0 & 0 \\
    \cdots & \cdots & \cdots & \ddots & \cdots & \cdots \\
    0 & 0 & 0 & \cdots & a & c \\
    0 & 0 & 0 & \cdots & b & a
    \end{array}
    \end{pmatrix}.$$

For example, we have
$$ T_1(a,b,c)=(a),\
    T_2(a,b,c)=\begin{pmatrix}
    \begin{array}{cc}
    a & c \\
    b & a
    \end{array}
    \end{pmatrix},
T_3(a,b,c)=\begin{pmatrix}
    \begin{array}{ccc}
    a & c & 0  \\
    b & a & c  \\
    0 & b & a  \\
    \end{array}
    \end{pmatrix}.
    $$
Let $T'_n(a,b,c)$ be another $n\times n$ tridiagonal Toeplitz matrix over $\F_q$, defined as
$$T'_n(a,b,c)=\begin{pmatrix}
    \begin{array}{ccccccc}
    a & 0 & c & \cdots & 0 & 0 & 0 \\
    0 & a & 0 & \cdots & 0 & 0 & 0 \\
    b & 0 & a & \cdots & 0 & 0 & 0 \\
    \cdots & \cdots & \cdots & \ddots & \cdots & \cdots & \cdots \\
    0 & 0 & 0 & \cdots & a & 0 & c \\
    0 & 0 & 0 & \cdots & 0 & a & 0 \\
    0 & 0 & 0 & \cdots & b & 0 & a
    \end{array}
    \end{pmatrix}.$$

For example, we have
$$ T'_1(a,b,c)=(a), \
    T'_2(a,b,c)=\begin{pmatrix}
    \begin{array}{cc}
    a & 0 \\
    0 & a
    \end{array}
    \end{pmatrix},
    T'_3(a,b,c)=\begin{pmatrix}
    \begin{array}{ccc}
    a & 0 & c  \\
    0 & a & 0  \\
    b & 0 & a  \\
    \end{array}
    \end{pmatrix}.$$

\begin{lemma}{\rm \cite{LCD-T-matric}}\label{lemma-eigenvalue of A}
Let $A$ be an $n\times n$ matrix over $\F_q$. We have the following cases:
\begin{itemize}
  \item char$\F_q$ is even: $-1$ is an eigenvalue of $A^2$ if and only if $-1$ is an eigenvalue of $A$.
  \item char$\F_q$ is odd: $-1$ is an eigenvalue of $A^2$ if and only if $-\mu$ or $\mu$ is an eigenvalue of $A$, where $\mu \in \F_{q^2}$ with $\mu^2=-1$.
\end{itemize}
\end{lemma}

In \cite{LCD-T-matric}, the eigenvalues of $T_n(a,b,b)$ have been characterized by the factorization of Dickson polynomials \cite{Factoring-DK-Po-}.
Now, we generalize this result.
Let
$$\phi_n(\lambda)=\det(T_n(a,b,c)-\lambda I_n),\ n\geq 1.$$
Similar to \cite[Lemma 2.2]{LCD-T-matric}, we have
$$\phi_n(\lambda)=(a-\lambda)\phi_{n-1}(\lambda)-bc\phi_{n-2}(\lambda).$$
For any $\alpha \in \F_q$ and $n\geq 0$, the Dickson polynomial of the second kind
$E_n(x,\alpha)\in \F_q[x]$
is defined recursively
$$E_n(x,\alpha)=xE_{n-1}(x,\alpha)-\alpha E_{n-2}(x,\alpha)$$
with the initial conditions $E_1(x,\alpha)=x$ and $E_0(x,\alpha)=1.$
Then we can get

\begin{prop}\label{prop-DD}
Keep the above notation, we have
$$\det(T_n(a,b,c)-\lambda I_n)=\phi_n(\lambda)=E_n(a-\lambda,bc),$$
for all $n\geq 1$ and $\lambda\in \overline{\F}_q$, where $\overline{\F}_q$ is an algebraic closure of $\F_q$.
\end{prop}

Then we recall the classical factorization of Dickson polynomials of the second kind over finite fields from \cite[Theorem 4]{Factoring-DK-Po-}.

\begin{theorem}\label{Theorem-E_n(x)-2}
Let $n\geq 1$ be an integer. Assume that {\rm char}$\F_q=p$, $n+1=p^r(m+1)$ and $\gcd(m+1,p)=1$. Then we have the following results.
\begin{itemize}
  \item If {\rm char}$\F_q$ is even, we have
  $$E_n(x,bc)=E_m(x,bc)^{2^r}x^{2^r-1},
  E_m(x,bc)=\prod_{i=1}^{m/2}(x-(bc)^{1/2}(\theta^i+\theta^{-i}))^2,$$
  where $\theta$ is a primitive $(m+1)$-th root of $1$.
  \item If {\rm char}$\F_q$ is odd, then we have
  $$E_n(x,bc)=E_m(x,bc)^{p^r}(x-2(bc)^{1/2})^{\frac{p^r-1}{2}}
  (x+2(bc)^{1/2})^{\frac{p^r-1}{2}},$$
  $$E_m(x,bc)=\prod_{i=1}^{m}(x-(bc)^{1/2}(\theta^i+\theta^{-i})),$$
  where $\theta$ is a primitive $2(m+1)$-th root of $1$.
\end{itemize}
\end{theorem}

Let
$$\ \phi'_n(\lambda)=\det(T'_n(a,b,c)-\lambda I_n),\ n\geq 1.$$
Similar to \cite{phi'}, in finite fields,
it can be proved that
\begin{eqnarray}
\phi'_{2n}(\lambda)=\phi_n(\lambda)^2,\  \phi'_{2n+1}(\lambda)=\phi_n(\lambda)\phi_{n+1}(\lambda).
\end{eqnarray}

\begin{remark}
According to Proposition \ref{prop-DD} and Theorem \ref{Theorem-E_n(x)-2},
the eigenvalues of $T_n(a,b,c)$ are obtained.
Then it is not difficult to get the eigenvalues of $T'_n(a,b,c)$ from (1).
\end{remark}

\subsection{Formally self-dual codes}

Firstly, we prove that Proposition 1 in \cite{isodual} is valid with respect to the Hermitian inner product.

\begin{prop}\label{isodual}
Let $A$ be an $n\times n$ matrix satisfying $A^T = QAQ$, where
$Q$ a monomial matrix such that $Q^2=I_n$, $I_n$ is identity matrix of size $n\times n$. Then the code $C = (I_n\ A)$ is a formally self-dual code of length $2n$ with respect to the Euclidean inner product and the Hermitian inner product.
\end{prop}

\begin{proof}
We only consider the Hermitian inner product over $\F_{q^2}$.
Since $Q$ is a monomial matrix of order $2$, $\overline{Q}=Q$.
The generator matrix of $C^{\perp_H}$ is $H=(-\overline{A}^T\ I_n)$. Using the hypothesis,
we know $\overline{A}^T=\overline{Q}\overline{A}\overline{Q}=Q\overline{A}Q$.
Then we have
$QH\mathcal{Q}=(I_n\ \overline{A})$, where
$\mathcal{Q}=\begin{pmatrix}
\begin{array}{cc}
0 & -Q\\
Q & 0
\end{array}
\end{pmatrix}.$
Hence the code generated by $(I_n\ \overline{A})$ is equivalent to $C^{\perp_H}$.
It is easy to see that the code with generator matrix $(I_n\ \overline{A})$ has the same weight distribution as the code with generator matrix $(I_n\ A)$. Therefore, the result follows.
\end{proof}

\begin{theorem}\label{theorem-Thm-FSD}
For any Toeplitz matrix $A$ and $f(x)\in \F_q[x]$, then the linear code $C$ with the generator matrix $(I\ f(A))$ is formally self-dual with respect to the Euclidean inner product and the Hermitian inner product.
\end{theorem}

\begin{proof}
From \cite[Theorem 1]{double-T-codes}, we know that there is a monomial matrix $Q$ with $Q^2=I_n$ such that $A^T = QAQ$. Then we have
$$(A^m)^{T}=(A^{T})^m=(QAQ)^m=QA^mQ.$$
Assume that $f(x)=\sum_{i=0}^{n-1}a_ix^i$. Then we have
$$f(A)^T=\left(\sum_{i=0}^{n-1}a_iA^i\right)^T=\sum_{i=0}^{n-1}a_i(A^i)^T=
Q\left(\sum_{i=0}^{n-1}a_iA^i\right)Q
=Qf(A)Q.$$
According to Proposition \ref{isodual}, the result follows.
\end{proof}

\section{Linear codes with small hulls}

\begin{definition}\label{definition}
Let $C$ be the $\F_q$-linear code of length $tn$ and dimension $n$ whose generator matrix is the $n\times tn$ matrix given by
$$(I_n \ f_1(A)\ f_2(A)\ \cdots \ f_{t-1}(A)),$$
where $I_n$ is identity matrix of order $n$, $A$ is a square matrix of order $n$ and $f_i(x)\in \F_q[x]$. We call $C$ a derivative code of index $t$ associated with $A$.
\end{definition}

\begin{definition}
If a square matrix $A$ over $\F_q$ satisfies $A^T = A$, then $A$ is a symmetric matrix. If a square matrix $A$ over $\F_{q^2}$ satisfies $\overline{A}^T = A$, then $A$ is a Hermitian matrix.
\end{definition}

\subsection{Constructions of codes with small Euclidean hulls}

\begin{theorem}\label{theorem-E-LCD of index t}
Assume that $\lambda_1,\lambda_2,\ldots,\lambda_n$ are eigenvalues of the symmetric matrix $A$.
Let $C$ be a linear code over $\F_q$ with the generator matrix $(I_n\ f_1(A)\ f_2(A)\ \cdots\ f_{t-1}(A))$, where $f_j(x)\in \F_q[x],\ 1\leq j\leq t-1$. Then we have the following results.
\begin{enumerate}
  \item [(1)] If $char \F_q$ is even, then $C$ is Euclidean LCD if and only if $$1\notin \{ f_1(\lambda_i)+f_2(\lambda_i)+\cdots+f_{t-1}(\lambda_i),\  1\leq i\leq n\}.$$
  \item [(2)] If $char \F_q$ is odd, then $C$ is Euclidean LCD if and only if
  $$-1 \notin \{ f_1(\lambda_i)^2+f_2(\lambda_i)^2+\cdots+f_{t-1}(\lambda_i)^2,\  1\leq i\leq n\}.$$
  In particular, when $t=2$, $C$ is Euclidean LCD if and only if
  $$\mu \notin \{ f_1(\lambda_i),\  1\leq i\leq n\} \cup \{-f_1(\lambda_i),\  1\leq i\leq n\},$$ where $\mu\in \F_{q^2}$ with $\mu^2=-1$.
\end{enumerate}
\end{theorem}

\begin{proof}
For any $f(x)\in \F_q[x]$,
it is not difficult to verify that $f(A)$ is symmetric if $A$ is symmetric.
And $f(\lambda_1),f(\lambda_2),\ldots,f(\lambda_n)$ are eigenvalues of the symmetric matrix $f(A)$.
It is well-known that $C$ is Euclidean LCD if and only if $GG^T$ is invertible (see \cite{LCD-Massey}), where $G=(I_n\ f_1(A)\ f_2(A)\ \cdots\ f_{t-1}(A))$.
In other words, $0$ is not an eigenvalue of $GG^T$.
Note that
$$
GG^{T}
=I_{n}+f_1(A)^2+f_2(A)^2+\cdots+f_{t-1}(A)^2,
$$
where $f_j(A)\ (1\leq j\leq t-1)$ is symmetric.
Hence $C$ is Euclidean LCD if and only if $-1$ is not an eigenvalue of $f_1(A)^2+f_2(A)^2+\cdots+f_{t-1}(A)^2$.
\begin{enumerate}
  \item [(1)] If $char \F_q$ is even, $f_1(A)^2+f_2(A)^2+\cdots+f_{t-1}(A)^2=(f_1(A)+f_2(A)+\cdots+f_{t-1}(A))^2$.
      According to Lemma \ref{lemma-eigenvalue of A}, we obtain the desired result.
  \item [(2)] If $char \F_q$ is odd, we just need to prove the case of $t=2$.
  According to Lemma \ref{lemma-eigenvalue of A}, we obtain the desired result.
\end{enumerate}
\end{proof}

\begin{theorem}\label{Thm-one-hull-E of index t}
Assume that $\lambda_1,\lambda_2,\ldots,\lambda_n$ are eigenvalues of the symmetric matrix $A$.
Let $C$ be a linear code over $\F_q$ with the generator matrix $(I_n\ f_1(A)\ f_2(A)\ \cdots
\ f_{t-1}(A))$, where $f_j(x)\in \F_q[x],\ 1\leq j\leq t-1$,
then we have the following results.
\begin{enumerate}
  \item [(1)] Assume that $char \F_q$ is even. If there is a unique $\lambda_i\ (1\leq i\leq n)$ such that $f_1(\lambda_i)+f_2(\lambda_i)+\cdots+f_{t-1}(\lambda_i)=1$, then $C$ is a linear $[tn,n]$ code with one-dimension Euclidean hull.
  \item [(2)] Assume that $char \F_q$ is odd. If there is a unique $\lambda_i\ (1\leq i\leq n)$ such that $f_1(\lambda_i)^2+f_2(\lambda_i)^2+\cdots+f_{t-1}(\lambda_i)^2=-1$, then $C$ is a linear $[tn,n]$ code with one-dimension Euclidean hull.
      In particular, when $t=2$, if there is a unique $\lambda_i\ (1\leq i\leq n)$ such that $f_1(\lambda_i)=\mu$ or $-\mu$, then $C$ is a linear $[2n,n]$ code with one-dimension Euclidean hull,
   where $\mu\in \F_{q^2}$ with $\mu^2=-1$.
\end{enumerate}
\end{theorem}

\begin{proof}
Combining the fact that $f(\lambda_1),f(\lambda_2),\ldots,f(\lambda_n)$ are eigenvalues of the symmetric matrix $f(A)$. According to \cite[Theorem 1]{Li-one-hull} and Lemma \ref{lemma-eigenvalue of A}, we obtain the desired result.
\end{proof}

\subsection{Constructions of codes with small Hermitian hulls}

\begin{theorem}\label{theorem-H-LCD of index t}
Assume that $\lambda_1,\lambda_2,\ldots,\lambda_n$ are eigenvalues of a Hermitian matrix $A$.
Let $C$ be a linear code over $\F_{q^2}$ with the generator matrix $(I_n\ f_1(A)\ f_2(A)\ \cdots\ f_{t-1}(A))$, where $f_j(x)\in \F_{q^2}[x]\ (1\leq j\leq t-1)$,
then we have the following results.
\begin{enumerate}
  \item [(1)] If $char \F_{q^2}$ is even, then $C$ is Hermitian LCD if and only if
  $$1\notin \{ f_1(\lambda_i)+f_2(\lambda_i)+\cdots+f_{t-1}(\lambda_i),\  1\leq i\leq n\}.$$
  \item [(2)] If $char \F_{q^2}$ is odd, then $C$ is Hermitian LCD if and only if
  $$-1 \notin \{ f_1(\lambda_i)^2+f_2(\lambda_i)^2+\cdots+f_{t-1}(\lambda_i)^2,\  1\leq i\leq n\}.$$
  In particular, when $t=2$, $C$ is Hermitian LCD if and only if
  $$\mu \notin \{ f_1(\lambda_i),\  1\leq i\leq n\} \cup \{-f_1(\lambda_i),\  1\leq i\leq n\},$$ where $\mu\in \F_{q^2}$ with $\mu^2=-1$.
\end{enumerate}
\end{theorem}

\begin{proof}
For any $f(x)\in \F_q[x]$,
it is not difficult to verify that $f(A)$ is a Hermitian matrix if $A$ is a Hermitian matrix.
Similar to Theorem \ref{theorem-E-LCD of index t}, the main difference is that we use $C$ is Hermitian LCD if and only if $G\overline{G}^T$ is invertible.
\end{proof}

\begin{lemma}\label{one-hull-Her}
Let $C$ be an $[n, k]$ linear code over $\F_{q^2}$ with
generator matrix $G = (I_k\ P)$. Then the code C has one-dimensional Hermitian hull if the matrix $P \overline{P}^T$ has an eigenvalue $-1$
with (algebraic) multiplicity $1$.
\end{lemma}
\begin{proof}
The proof is similar to the proof of \cite[Theorem 1]{Li-one-hull}, the difference is that we use $C$ is Hermitian LCD if and only if $G\overline{G}^T$ is invertible (see \cite{LCD-6}).
\end{proof}

\begin{theorem}\label{Thm-one-hull-H of index t}
Assume that $\lambda_1,\lambda_2,\ldots,\lambda_n$ are a eigenvalues of Hermitian matrix $A$.
Let $C$ be a linear code over $\F_q$ with the generator matrix $(I_n\ f_1(A)\ f_2(A)\ \cdots
\ f_{t-1}(A))$, where $f_j(x)\in \F_q[x],\ 1\leq j\leq t-1$,
then we have the following results.
\begin{enumerate}
  \item [(1)] Assume that $char \F_q$ is even. If there is a unique $\lambda_i\ (1\leq i\leq n)$ such that $f_1(\lambda_i)+f_2(\lambda_i)+\cdots+f_{t-1}(\lambda_i)=1$, then $C$ is a linear $[tn,n]$ code with one-dimension Hermitian hull.
  \item [(2)] Assume that $char \F_q$ is odd. If there is a unique $\lambda_i\ (1\leq i\leq n)$ such that $f_1(\lambda_i)^2+f_2(\lambda_i)^2+\cdots+f_{t-1}(\lambda_i)^2=-1$, then $C$ is a linear $[tn,n]$ code with one-dimension Hermitian hull.
      In particular, when $t=2$, if there is a unique $\lambda_i\ (1\leq i\leq n)$ such that $f(\lambda_i)=\mu$ or $-\mu$, then $C$ is a linear $[2n,n]$ code with one-dimension Hermitian hull,
   where $\mu\in \F_{q^2}$ with $\mu^2=-1$.
\end{enumerate}
\end{theorem}

\begin{proof}
Combining the fact that $f(\lambda_1),f(\lambda_2),\ldots,f(\lambda_n)$ are eigenvalues of the Hermitian matrix $f(A)$. Acording Lemma \ref{one-hull-Her} and Lemma \ref{lemma-eigenvalue of A}, we obtain the desired result.
\end{proof}

\section{Codes with small Euclidean hulls from Toeplitz matrices}

The eigenvalues of matrices $T_n(a,b,b)$ and $T'_n(a,b,b)$ are given in Section $2$.
Then we have following results according to Theorem \ref{theorem-Thm-FSD} and Theorem \ref{theorem-E-LCD of index t}.
In addition, let $C$ be the linear code defined as in in Definition \ref{definition}. When $A=T_n(a,b,b)$ or $T'_n(a,b,b)$ and $t=2$, $C$ is FSD by Theorem \ref{theorem-Thm-FSD}.

\subsection{Euclidean LCD codes}

\begin{theorem}\label{E-even-0}
Let $a,b\in \F_q$ with $b\neq 0$ and $n\geq 2$ be an integer. Let $C$ be the linear code over $\F_q$ with the generator matrix $(I_n\ f_1(T_n(a,b,b))\ \ldots\ f_{t-1}(T_n(a,b,b)))$, where $f_j(x)\in \F_{q}[x]\ (1\leq j\leq t-1)$. Assume that ${\rm char}\F_q$ is even. Let $r$ be the largest integer such that $n+1=2^r(m+1)$, then $m$ is even. Then we have
\begin{itemize}
  \item [(1)] If $r=0$, then
  $C$ is Euclidean LCD if and only if $$1\notin \{ f_1(a-b(\theta^i+\theta^{-i}))+ \cdots+f_{t-1}(a-b(\theta^i+\theta^{-i})),\  1\leq i\leq \frac{n}{2}\},$$
  where $\theta$ is a primitive $(n+1)$-th root of $1$.
  \item [(2)] If $r\geq 1$ and $m\geq 1$, then $C$ is Euclidean LCD if and only if
$$1 \notin \left\{f_1(a)+\cdots+f_{t-1}(a) \right\}$$ $$\cup \left\{f_1(a+b(\theta^i+\theta^{-i}))+\cdots+f_{t-1}(a+b(\theta^i+\theta^{-i})): 1\leq i\leq \frac{m}{2}\right\},$$
where $\theta$ is a primitive $(m+1)$-th root of $1$.
  \item [(3)] If $r\geq 1$ and $m=0$, then $C$ is Euclidean LCD if and only if
  $$1 \notin \left\{f_1(a)+\cdots+f_{t-1}(a) \right\}.$$
\end{itemize}
\end{theorem}

\begin{proof}
By Proposition \ref{prop-DD}, we know that
$\det(T_n(a,b,b)-\lambda I_n)=\phi(\lambda)=E_n(a-\lambda,b^2).$

(1) If $r=0$, then $n$ is even.
According to Theorem \ref{Theorem-E_n(x)-2}, we obtain that the eigenvalues of $T_n(a,b,b)$ are
$\{{\bf 2}\cdot (a-b(\theta^i+\theta^{-i})),\  1\leq i\leq \frac{n}{2}\}.$
Using Theorem \ref{theorem-E-LCD of index t}, we complete the proof.

(2) If $r\geq 1$ and $m\geq 1$. According to Theorem \ref{Theorem-E_n(x)-2}, we obtain that the eigenvalues of $T_n(a,b,b)$ are
$\{{\bf 2^r-1}\cdot a\}\cup \{{\bf 2^{r+1}}\cdot (a-b(\theta^i+\theta^{-i})),\  1\leq i\leq \frac{m}{2}\}.$
Using Theorem \ref{theorem-E-LCD of index t}, we complete the proof.

(3) If $r\geq 1$ and $m= 0$. According to Theorem \ref{Theorem-E_n(x)-2}, we obtain that the eigenvalues of $T_n(a,b,b)$ are
$\{{\bf 2^r-1}\cdot a\}.$
Using Theorem \ref{theorem-E-LCD of index t}, we complete the proof.
\end{proof}

\begin{example}
Let $\F^*_{2^{12}}=\langle \omega \rangle$ and $f_1(x)=x^9+x^8+x^4+x^3+x^2\in \F_2[x]$. Consider the linear code $C$ over $\F_2$ with the generator matrix
$(I_{12}\ f_1(T_{12}(1,1,1))).$ So $a=b=c=1,\ n=12,\ t=2,\ r=0$.
According to (1) of Theorem \ref{E-even-0}. Let $\theta=\omega^{315}$ such that $\theta^{13}=1$.
Let $$S=\{ f_1(a+b(\theta^i+\theta^{-i})):\  1\leq i\leq \frac{n}{2}\}.$$
Computing by Magma \cite{magma},
we obtain that
$$S= \{\omega^{2015},\omega^{3055},\omega^{3575},\omega^{3835},
\omega^{3965},\omega^{4030} \},$$
and $C$ has parameters $[24,12,6]$.
Hence $1\notin S$. Therefore, $C$ is a Euclidean LCD code from Theorem \ref{E-even-0}, and it is optimal (see \cite{AH-BLCD-17-24}).
\end{example}

\begin{remark}
It is easy to see that the code with the generator matrix $(I_{12}\ T_{12}(1,1,1))$ is a binary FSD Euclidean LCD $[24,12,3]$ code. A binary optimal Euclidean LCD $[24,12,6]$ code is obtained by Theorem \ref{E-even-0}, and it is still FSD.
Therefore, compared with \cite{LCD-T-matric}, we improve the largest minimum distance.
\end{remark}

\begin{remark}
In Tables 1 and 2,
we collect some optimal, almost optimal or best-known binary Euclidean LCD codes with respect to the newest codetable for binary Euclidean LCD codes (see \cite{B-LCD-40}), where $``*"$ denotes optimal binary Euclidean LCD codes, $``\star"$ denotes almost optimal binary Euclidean LCD codes, $``\diamond"$ denotes the best-known binary Euclidean LCD codes.
\end{remark}

Throughout this paper, let $d^E_2 (n, k)$ denote the largest minimum distance among all Euclidean LCD $[n, k]$ codes over $\F_2$.
The following lemma was proved in \cite[Theorem 8]{2-lcd} for binary LCD codes.
\begin{lemma}\label{lemma-leq}
Suppose that $2\leq k\leq n$. Then
$$d^E_2(n,k)\leq d^E_2(n,k-1).$$
\end{lemma}

The following lemma was proved in \cite[Corollary 4]{B-LCD-40}.
\begin{lemma}\label{lemma-leq-2}
(1) If $k$ is odd, then
$d^E_2(n,k)\leq d^E_2(n-1,k-1).$
(2) If $k$ is even and $d_2^E(n,k)$ is odd, then $d_2^E(n+1,k)\geq d_2^E(n,k)+1$.
\end{lemma}

\begin{cor}\label{[n,k,7]}
For $(n,k)\in \{(33,16),(36,17),(38,17),(38,18)\}$, there is a binary LCD $[n,k,7]$ code.
\end{cor}

\begin{proof}
There exist binary LCD $[34,17,7]$, $[36,18,7]$ and $[38,19,7]$ codes from Table 1. By Lemma \ref{lemma-leq}, there exist binary LCD $[36,17,7]$, $[38,17,7]$ and $[38,18,7]$ codes. Using (1) of Lemma \ref{lemma-leq-2}, there exists a binary LCD $[33,16,7]$ code.
\end{proof}

\begin{cor}\label{[n,k,8]}
For $(n,k)\in \{(34,16),(37,18),(39,18),(40,17),(40,18),(40,19)\}$, there is a binary LCD $[n,k,8]$ code.
\end{cor}

\begin{proof}
There exists a binary LCD $[40,20,8]$ code from Table 1. By Lemma \ref{lemma-leq}, there exists a binary LCD $[n,k,8]$ code for $(n,k)\in \{(40,17),(40,18),(40,19)\}$. Using (1) of Lemma \ref{lemma-leq-2}, there exists a binary LCD $[39,18,8]$ code. Combining Corollary \ref{[n,k,7]} and (2) of Lemma \ref{lemma-leq-2}, there exist binary $[34,16,8]$ and $[37,18,8]$ codes.
\end{proof}
\begin{remark}
Compared with Table 1 and Table 2 in \cite{B-LCD-40}, the codes of Corollary \ref{[n,k,8]} and Corollary \ref{[n,k,7]} have better parameters than its parameters.
For example, the linear code of the
length 37 with the dimension 18 in \cite{B-LCD-40} has the minimal distance 6, while the LCD code of the length 37 with the dimension 18 we obtained has the minimal distance 8.
In Table 1 and Table 2, we list some optimal, almost optimal or the best-known binary LCD codes. These codes improve the previously known lower bounds on the largest minimum weights in \cite{B-LCD-40}. It is worth mentioning that the codes of Table 1 are FSD. That is to say, we use FSD LCD codes to improve the lower bound of LCD codes on the largest minimum distance.
They are likely to be optimal FSD binary LCD codes.
\end{remark}

\begin{center}Table 1: Binary Euclidean FSD LCD codes\\

\begin{tabular}{|c|c|c|c|}
  \hline
A & f(x) & Parameters & Best know in Ref. \\
\hline
  $T_3(1,1,1)$ & $x + 1$ & $[6,3,2]^*$ & [6,3,2] \cite{bound-13} \\
\hline
$T_4(1,1,1)$ & $x$ & $[8,4,3]^*$& [8,4,3] \cite{bound-13}\\
\hline
$T_5(1,1,1)$ & $x^2+x$ & $[10,5,3]^*$& [10,5,3] \cite{bound-13}\\
\hline
$T_6(1,1,1)$ & $x^3 $& $[12,6,4]^*$& [12,6,4] \cite{bound-13}\\
\hline
$T_7(1,1,1)$ & $x^2+x $& $[14,7,4]^*$& [14,7,4] \cite{HS-BLCD-1-16}\\
\hline
$T_8(1,1,1) $& $x^2+x $& $[16,8,4]^\star$& [16,8,5] \cite{HS-BLCD-1-16}\\
\hline
$T_9(1,1,1)$ & $x^8+x^5 $& $[18,9,4]^\star$& [18,9,5] \cite{AH-BLCD-17-24}\\
\hline
$T_{10}(1,1,1)$ & $x^5+x^4+x^3 $& $[20,10,5]^\star$& [20,10,6] \cite{AH-BLCD-17-24}\\
\hline
$T_{11}(1,1,1)$ & $x^9+x^8+x^7+x$ & $[22,11,5]^\star$& [22,11,6] \cite{AH-BLCD-17-24}\\
\hline
$T_{12}(1,1,1)$ &$ x^9+x^8+x^4+x^3+x^2 $& $[24,12,6]^*$& [24,12,6] \cite{AH-BLCD-17-24}\\
\hline
$T_{13}(1,1,1) $& $x^{12}+x^{11}+x^7+x^4 $& $[26,13,6]^\star$& [26,13,7] \cite{2-LCD-30}\\
\hline
$T_{14}(1,1,1)$ & $x^{12}+x^{11}+x $& $[28,14,6]^\star$& [28,14,7] \cite{2-LCD-30}\\
\hline
$T_{15}(1,1,1)$ & $x^{14} + x^{13} + x^{12} + x^{10}$ & $[30,15,6]^\star$& [30,15,7] \cite{2-LCD-30}\\
\hline
$T_{16}(1,1,1)$ &$ x^{15} + x^{12}$ & $[32,16,6]^\diamond$& [32,16,6] \cite{B-LCD-40}\\
\hline
$T_{17}(1,1,1)$ &$ x^{14}+x^{10}+x^7+x$ & $[34,17,7]^\diamond$& [34,17,6] \cite{B-LCD-40}\\
\hline
$T_{18}(1,1,1) $&$ x^{15}+x^{13}+x^{11}+x^{10}+x^9 $& $[36,18,7]^\diamond$& [36,18,6] \cite{B-LCD-40}\\
\hline
$T_{19}(1,1,1)$ &$ x^{15}+x^{13}+x^{10}+x^9$ & $[38,19,7]^\diamond$& [38,18,6] \cite{B-LCD-40}\\
\hline
$T_{20}(1,1,1)$ &$ x^{13}+x^{11}+x^8+x^5+x^4+x^2+x$ & $[40,20,8]^\diamond$& [40,20,6] \cite{B-LCD-40}\\
\hline
$T_{25}(1,1,1)$ &$ x^{15}+x^{11}+x^{10}+x^8+x^7+x$ &$ [50,25,9]^\diamond$& \\
\hline
\end{tabular}
\end{center}

\begin{center}Table 2: Binary Euclidean LCD codes
\begin{tabular}{|c|c|c|c|}
\hline
$A$ & $f_1(x)$ & $f_2(x)$ & Parameters \\
\hline
$T_3(1,1,1)$ & $x^2$ & $x$ & $[9,3,4]^*$\\
\hline
$T_4(1,1,1)$ &$ x^2$ & $x^2+x $& $[12,4,5]^*$\\
\hline
$T_5(1,1,1)$ & $x^3 + x^2 + 1$ & $x^3 + x + 1$ & $[15,5,6]^*$\\
\hline
$T_6(1,1,1)$ &$ x^5$ & $x^5 + x^4 + x^3$ & $[18,6,6]^\star$\\
\hline
$T_7(1,1,1)$ & $x^5$ & $x^5 + x^4 + x^3 $& $[21,7,7]^\star$\\
\hline
$T_8(1,1,1)$ &$ x^5+1$ &  $x^4 + x+1 $& $[24,8,8]^*$\\
\hline
$T_9(1,1,1)$ &$ x^5 $&$ x^5 + x^4 + x^3$ & $[27,9,8]^\star$\\
\hline
$T_{10}(1,1,1)$ & $x^5$ & $x^5 + x^4 + x^2+x $& $[30,10,9]^\star$\\
\hline
\end{tabular}
\end{center}

\begin{theorem}\label{E-odd-0}
Let $a,b\in \F_q$ with $b\neq 0$ and $n\geq 2$ be an integer. Let $C$ be the linear code over $\F_q$ with the generator matrix $(I_n\ f_1(T_n(a,b,b))\ \cdots\ f_{t-1}(T_n(a,b,b)))$, where $f_j(x)\in \F_{q}[x]\ (1\leq j\leq t-1)$. Assume that ${\rm char}\F_q=p$ is odd. Let $r$ be the largest integer such that $n+1=p^r(m+1)$. Then we get the following results.
\begin{itemize}
  \item [(1)] If $r=0$, then
  $C$ is Euclidean LCD if and only if $$-1\notin \{ f_1(a-b(\theta^i+\theta^{-i}))^2+ \cdots+f_{t-1}(a-b(\theta^i+\theta^{-i}))^2,\  1\leq i\leq n\},$$
  where $\theta$ is a primitive $2(n+1)$-th root of $1$.
  In particular, when $t=2$, $C$ is Euclidean LCD if and only if $$\mu \notin \{ f_1(a-b(\theta^i+\theta^{-i})):\  1\leq i\leq n\}\cup \{ -f_1(a-b(\theta^i+\theta^{-i})):\  1\leq i\leq n\},$$
  where $\mu\in \F_{q^2}$ with $\mu^2=-1$ and $\theta$ is a primitive $2(n+1)$-th root of $1$.
  \item [(2)] If $r\geq 1$ and $m\geq 1$, then $C$ is Euclidean LCD if and only if
  $$-1 \notin \left\{f_1(a+ 2b)^2+\cdots+f_{t-1}(a+ 2b)^2, f_1(a- 2b)^2+\cdots+f_{t-1}(a-2b)^2 \right\}$$ $$\cup \left\{f_1(a-b(\theta^i+\theta^{-i}))^2+\cdots+f_{t-1}(a-b(\theta^i+\theta^{-i}))^2: 1\leq i\leq m\right\},$$
  where $\theta$ is a primitive $2(m+1)$-th root of $1$.
  In particular, when $t=2$, $C$ is Euclidean LCD if and only if
  $$\mu \notin \{f_1(a+ 2b),-f_1(a+ 2b),f_1(a-2b),-f_1(a-2b)\}$$ $$\cup
  \{f_1(a-b(\theta^i+\theta^{-i})): 1\leq i\leq m\},$$
  where $\mu\in \F_{q^2}$ with $\mu^2=-1$ and $\theta$ is a primitive $2(m+1)$-th root of $1$.
  \item [(3)] If $r\geq 1$ and $m=0$, then $C$ is Euclidean LCD if and only if
  $$-1 \notin \left\{f_1(a+ 2b)^2+\cdots+f_{t-1}(a+ 2b)^2, f_1(a- 2b)^2+\cdots+f_{t-1}(a-2b)^2 \right\}.$$
  In particular, when $t=2$, $C$ is Euclidean LCD if and only if
  $$\mu \notin \{f_1(a+ 2b),-f_1(a+ 2b),f_1(a-2b),-f_1(a-2b)\},$$
  where $\mu\in \F_{q^2}$ with $\mu^2=-1$.
\end{itemize}
\end{theorem}

\begin{proof}
By Proposition \ref{prop-DD}, we know that
$\det(T_n(a,b,b)-\lambda I_n)=\phi(\lambda)=E_n(a-\lambda,b^2).$

(1) If $r=0$, then $n$ is even.
According to Theorem \ref{Theorem-E_n(x)-2}, we obtain that the eigenvalues of $T_n(a,b,b)$ are
$\{a-b(\theta^i+\theta^{-i}),\  1\leq i\leq n\}.$
Using Theorem \ref{theorem-E-LCD of index t}, we complete the proof.

(2) If $r\geq 1$ and $m\geq 1$. According to Theorem \ref{Theorem-E_n(x)-2}, we obtain that the eigenvalues of $T_n(a,b,b)$ are
$\{{\bf \frac{p^r-1}{2}}\cdot (a+2b)\}\cup \{{\bf \frac{p^r-1}{2}}\cdot (a-2b)\}\cup \{{\bf 2^{r+1}}\cdot (a-b(\theta^i+\theta^{-i})),\  1\leq i\leq m\}.$
Using Theorem \ref{theorem-E-LCD of index t}, we complete the proof.

(3) If $r\geq 1$ and $m= 0$. According to Theorem \ref{Theorem-E_n(x)-2}, we obtain that the eigenvalues of $T_n(a,b,b)$ are
$\{{\bf \frac{p^r-1}{2}}\cdot (a+2b)\}\cup \{{\bf \frac{p^r-1}{2}}\cdot (a-2b)\}.$
Using Theorem \ref{theorem-E-LCD of index t}, we complete the proof.
\end{proof}

\begin{remark}
In Tables 3 and 4,
we collect some optimal, almost optimal or best-known ternary Euclidean LCD codes with respect to the newest codetable for ternary Euclidean LCD codes (see \cite{AH-TLCD,
T-11-19}), where $``*"$ denotes optimal ternary Euclidean LCD codes, $``\star"$ denotes almost optimal ternary Euclidean LCD codes, $``\diamond"$ denotes the best-known ternary Euclidean LCD codes. For $n> 20$, the codes of Tables 3 and 4 are the best-known. It is worth mentioning that the codes of Table 3 are FSD.
\end{remark}

\begin{center}Table 3: Ternary FSD LCD codes
\begin{tabular}{|c|c|c|}
\hline
$A$ &$ f(x) $& Parameters \\
\hline
$T_3(1,1,1) $&$ x $&$ [6,3,3]^*$\\
\hline
$T_4(1,1,1) $&$ x^2+2 $&$ [8,4,4]^*$\\
\hline
$T_5(1,1,1) $&$ x^2 $&$ [10,5,4]^\star$\\
\hline
$T_6(1,1,1) $&$ x^5 $&$ [12,6,4]^\star$\\
\hline
$T_7(1,1,1) $&$ 2x^5 + x^3 + x^2 $&$ [14,7,5]^\star$\\
\hline
$T_8(1,1,1) $&$ 2x^7 + x^2 $&$ [16,8,5]^\star$\\
\hline
$T_9(1,1,1) $&$ 2x^8 + 2x^3 + x $&$ [18,9,6]^*$\\
\hline
$T_{10}(1,1,1) $&$ 2x^8 + x^6 + x^5 $&$ [20,10,6]^\star$\\
\hline
$T_{11}(1,1,1) $&$ x^{10} + x^8 + x^7 + x^6 $&$ [22,11,7]^\diamond$\\
\hline
$T_{12}(1,1,1) $&$ x^{10} + x^8 + x^7 + x^5  $&$ [24,12,7]^\diamond$\\
\hline
$T_{13}(1,1,1) $&$ 2x^{11} + x^9 + x^8  $&$ [26,13,7]^\diamond$\\
\hline
$T_{14}(1,1,1) $&$ x^{13}+x^{12}+2x^{11} + x^8 + x^7 + 2x^5  $&$ [28,14,8]^\diamond$\\
\hline
$T_{15}(1,1,1) $&$ 2x^{12} + x^{10} + 2x^9 + x^8  $&$ [30,15,8]^\diamond$\\
\hline
\end{tabular}
\end{center}

\begin{center} Table 4: Ternary LCD codes\\

\begin{tabular}{|c|c|c|c|}
\hline
$A $&$ f_1(x) $&$ f_2(x) $& Parameters \\
\hline
$T_2(1,1,1) $&$ x $&$ x+1 $&$ [6,2,4]^*$\\
\hline
$T_3(1,1,1) $&$ x^2 $&$ 2x^2+x $&$ [9,3,4]^\star$\\
\hline
$T_4(1,1,1) $&$ x^3 $&$ x^2 $&$ [12,4,6]^*$\\
\hline
$T_6(1,1,1) $&$ x^5 $&$ 2x^5+x^3 $&$ [18,6,8]^\star$\\
\hline
$T_7(1,1,1) $&$ x^6 $&$ 2x^5+x^3 $&$ [21,7,8]^\diamond$\\
\hline
$T_8(1,1,1) $&$ x^7+2x^6+x^2 $&$ 2x^7+2x^3+1 $&$ [24,8,10]^\diamond$\\
\hline
\end{tabular}
\end{center}

\begin{theorem}\label{C'-LCD}
 Let $a,b\in \F_q$ with $b\neq 0$ and $n\geq 1$ be an integer. Let $C'_{2n}(a,b)$ $(resp.\ C'_{2n+1}(a,b))$ be a linear code over $\F_q$ with the generator matrix $$(I_{2n}\ f_1(T'_{2n}(a,b,b))\ \cdots\ f_{t-1}(T'_{2n}(a,b,b)))$$
$$(resp.\ (I_{2n+1}\ f_1(T'_{2n+1}(a,b,b))\ \cdots\ f_{t-1}(T'_{2n+1}(a,b,b)))).$$
Let $C_n(a,b)$ $(resp.\ C_{n+1}(a,b))$ be a linear code with the generator matrix $$(I_n\ f_1(T_n(a,b,b))\ \cdots\ f_{t-1}(T_n(a,b,b)))$$ $$(resp.\ (I_{n+1}\ f_1(T_{n+1}(a,b,b))\ \cdots\ f_{t-1}(T_{n+1}(a,b,b)))).$$
Then we have the following results.
\begin{enumerate}
  \item [(1)] $C'_{2n}(a,b)$ is Euclidean LCD if and only if
$C_n(a,b)$ is Euclidean LCD.
  \item [(2)] $C'_{2n+1}(a,b)$ is Euclidean LCD if and only if
$C_n(a,b)$ and $C_{n+1}(a,b)$ are Euclidean LCD.
\end{enumerate}
\end{theorem}

\begin{proof}
(1) Let $G'=(I_{2n}\ f_1(T'_{2n}(a,b,b))\ \cdots\ f_{t-1}(T'_{2n}(a,b,b)))$, $C'_{2n}$ is Euclidean LCD if and only if $G'G'^{T}$ is invertible \cite{LCD-Massey}. Since
$$
G'G'^{T}=I_{2n}+f_1(T'_{2n}(a,b,b))^2+\cdots+f_{t-1}(T'_{2n}(a,b,b))^2,
$$
$C'_{2n}$ is Euclidean LCD if and only if $-1$ is not an eigenvalue of $f_1(T'_{2n}(a,b,b))^2+\cdots+f_{t-1}(T'_{2n}(a,b,b))^2$.
According to the relationship between eigenvalues of $T'_{2n}$ and $T_n$ (see (1) of Section 2),
$-1$ is not an eigenvalue of $f_1(T'_{2n}(a,b,b))^2+\cdots+f_{t-1}(T'_{2n}(a,b,b))^2$ if and only if $-1$ is not an eigenvalue of $f_1(T_{n}(a,b,b))^2+\cdots+f_{t-1}(T_{n}(a,b,b))^2$. In other words, $C'_{2n}(a,b)$ is Euclidean LCD if and only if
$C_n(a,b)$ is Euclidean LCD.

(2) The proof is similar to the proof of $(1)$, so we omit it here.
\end{proof}

\subsection{Codes with one-dimension Euclidean hull}

\begin{theorem}\label{E-1}
Let $a,b\in \F_q$ with $b\neq 0$ and $n\geq 2$ be an integer.
Let $C$ be the linear code over $\F_q$ with the generator matrix $(I_n\ f_1(T_n(a,b,b))\ \cdots\ f_{t-1}(T_n(a,b,b)))$, where $f_j(x)\in \F_q[x]\ (1\leq j\leq t-1)$. Assume that $char\F_q=p$.
Let $r$ be the largest integer such that $n+1=p^r(m+1)$.
Then we have the following results.
\begin{enumerate}
  \item [(1)] Assume that $char \F_q$ is even. When $r=1$ and $f_1(a)+f_2(a)+\cdots+f_{t-1}(a)=1$,
  but
  $$f_1(a-b(\theta^i+\theta^{-i}))+\cdots+f_{t-1}(a-b(\theta^i+\theta^{-i}))\neq 1,1\leq i \leq \frac{m}{2},$$
  where $\theta$ is a primitive $(m+1)$-th root of $1$,
  then $C$ is a linear $[tn,n]$ code with one-dimension Euclidean hull.
  \item [(2)] Assume that $char \F_q$ is odd.
  \begin{itemize}
    \item When $r=0$, if there is a unique $i\ (1\leq i\leq m)$ such that $$f_1(a-b(\theta^i+\theta^{-i}))^2+\cdots+f_{t-1}(a-b(\theta^i+\theta^{-i}))^2=-1,$$
        where $\theta$ is a primitive $2(m+1)$-th root of $1$.
        Then $C$ is a linear $[tn,n]$ code with one-dimension Euclidean hull.
      In particular, when $t=2$, if there is a unique $i\ (1\leq i\leq m)$ such that $f_1(a-b(\theta^i+\theta^{-i}))=\mu$ or $-\mu$, where $\theta$ is a primitive $2(m+1)$-th root of $1$ and $\mu\in \F_{q^2}$ with $\mu^2=-1$.
      Then $C$ is a linear $[2n,n]$ code with one-dimension Euclidean hull.

    \item When $p=3$ and $r=1$, if there is a unique $x$ such that $f_1(x)^2+\ldots+f_{t-1}(x)^2=-1$, where $x\in\{a-2b,a+2b\}$, but $$f_1(a-b(\theta^i+\theta^{-i}))^2+\cdots+f_{t-1}(a-b(\theta^i+\theta^{-i}))^2\neq -1,\ 1\leq i \leq m,$$
  where $\theta$ is a primitive $2(m+1)$-th root of $1$. Then $C$ is a linear $[tn,n]$ code with one-dimension Euclidean hull.
  In particular, when $t=2$, if there is a unique $x$ such that $f_1(x)=\mu$ or $-\mu$, where $x\in\{a-2b,a+2b\}$, but
  $$f_1(a-b(\theta^i+\theta^{-i}))\neq \mu\ or\ -\mu,\ 1\leq i \leq m,$$
  where $\theta$ is a primitive $2(m+1)$-th root of $1$ and $\mu\in \F_{q^2}$ with $\mu^2=-1$.
   Then $C$ is a linear $[2n,n]$ code with one-dimension Euclidean hull.
  \end{itemize}
\end{enumerate}
\end{theorem}

\begin{proof}
(1) When $char \F_q$ is even and $r=1$, by Theorem \ref{Theorem-E_n(x)-2}, we know that $T_n(a,b,b)$ has an eigenvalue $a$ with algebraic multiplicity $1$.
Using assumptions, it is easy to check that $f_1(T_n(a,b,b))+\cdots+f_{t-1}(T_n(a,b,b))$ has an eigenvalue $1$ with algebraic multiplicity $1$. Using the fact that $char \F_q$ is even and $T_n(a,b,b)$ is symmetric, we have that $f_1(T_n(a,b,b))^2+\cdots+f_{t-1}(T_n(a,b,b))^2$ has an eigenvalue $1$ with algebraic multiplicity $1$. Using Theorem 1 in \cite{Li-one-hull}, we complete the proof.

(2) The proof is similar to the proof of (1), so we omit it here.
\end{proof}

\begin{example}
Let $\F^*_{5^4}=\langle \omega \rangle$ and $f_1(x)=x^4+2x^3+3x^2+x\in \F_5[x]$. Let $C$ be the linear code over $\F_5$ with the generator matrix $(I_{7}\ f_1(T_{7}(1,1,1))),$ so $a=b=1,\ n=7,\ t=2,\ r=0$.
According to (2) of Theorem \ref{E-1}. Let $\theta=\omega^{39}$ such that $\theta^{16}=1$, and $\mu=2$ such that $\mu^2=-1$.
Let $$S=\{f_1(1-(\theta^i+\theta^{-i})):\ 1\leq i\leq n\}.$$
Computing by Magma \cite{magma}, we obtain that
$$S=\{\omega^{52},\omega^{260},2,\omega^{324},
\omega^{372},\omega^{564},\omega^{612} \},$$
and $C$ has parameters $[14,7,6]$. Hence $\mu=2\in S$ of algebraic multiplicity $1$ and $-\mu=3\notin S$. Therefore, from Theorem \ref{E-1}, $C$ is an optimal code with one-dimension Euclidean hull according to the Database \cite{codetables}.
\end{example}

\section{Codes with small Hermitian hulls}

The eigenvalues of the matrix $T_n(a,b,b^q)$ are given in Section $2$. When $a\in \F_q$ and $b\in \F_{q^2}$, $T_n(a,b,b^q)$ is a Hermitian matrix over $\F_{q^2}$, then we have the following results according to Theorem \ref{theorem-Thm-FSD} and Theorem \ref{theorem-H-LCD of index t}.
In addition, let $C$ be the linear code defined as in Definition \ref{definition}. When $A=T_n(a,b,b^q)$ and $t=2$, $C$ is FSD by Theorem \ref{theorem-Thm-FSD}.
\subsection{Hermitian LCD codes}

\begin{theorem}\label{H-even-0}
Let $a\in \F_q$ and $b\in \F_{q^2}$ with $b\neq 0$ and $n\geq 2$ be an integer. Let $C$ be the linear code over $\F_{q^2}$ with the generator matrix $(I_n\ f_1(T_n(a,b,b^q))\ \cdots\ f_{t-1}(T_n(a,b,b^q)))$, where $f_j(x)\in \F_{q^2}[x]\ (1\leq j\leq t-1)$. Assume that ${\rm char}\ \F_{q^2}$ is even. Let $r$ be the largest integer such that $n+1=2^r(m+1)$. Then we have the following results.
\begin{itemize}
  \item [(1)] If $r=0$, then
  $C$ is Hermitian LCD if and only if $$1\notin \{ f_1(a-b^{\frac{(q+1)q^2}{2}}(\theta^i+\theta^{-i}))+ \cdots+f_{t-1}(a-b^{\frac{(q+1)q^2}{2}}(\theta^i+\theta^{-i})),\  1\leq i\leq \frac{n}{2}\},$$
  where $\theta$ is a primitive $(n+1)$-th root of $1$.
  \item [(2)] If $r\geq 1$ and $m\geq 1$, then $C$ is Hermitian LCD if and only if
$$1 \notin \left\{f_1(a)+\cdots+f_{t-1}(a) \right\}\cup$$ $$ \left\{f_1(a-b^{\frac{(q+1)q^2}{2}}(\theta^i+\theta^{-i}))+\cdots+
f_{t-1}(a-b^{\frac{(q+1)q^2}{2}}(\theta^i+\theta^{-i})): 1\leq i\leq \frac{m}{2}\right\},$$
where $\theta$ is a primitive $(m+1)$-th root of $1$.
  \item [(3)] If $r\geq 1$ and $m=0$, then $C$ is Hermitian LCD if and only if
  $$1 \notin \left\{f_1(a)+\cdots+f_{t-1}(a) \right\}.$$
\end{itemize}
\end{theorem}

\begin{proof}
The proof is similar to the proof of Theorem \ref{E-even-0}, the main difference is that
we use the fact that $C$ with the generator matrix $G$ is Hermitian LCD if and only if $G\overline{G}^T$ is invertible. In addition, $T_q(a,b,b^q)$ is a Hermitian matrix. And $(b^{q+1})^{1/2}=(b^{(q+1)q^2})^{1/2}=b^{\frac{b^{(q+1)q^2}}{2}}$.
\end{proof}

\begin{example}
Let $\F^*_{4^{3}}=\langle \xi \rangle$, $\F^*_4=\langle \omega \rangle$ and $f_1(x)=\omega x^3+x\in \F_4[x]$. Let $C$ be the linear code over $\F_4$ with the generator matrix
$$
(I_{6},f_1(T_{6}(1,\omega^2,\omega)))=
\begin{pmatrix}
    \begin{array}{cccccccccccc}
    1 & 0 & 0&0 & 0 & 0 & 1 &  1& 1&\omega & 0 & 0\\
    0 & 1 & 0&0 & 0 & 0 & \omega & \omega^2 & \omega&1 & \omega & 0\\
    0 & 0 & 1&0 & 0 & 0 & \omega^2 & \omega^2 & \omega^2&\omega & 1 & \omega\\
    0 & 0 & 0&1 & 0 & 0 & \omega & \omega^2 & \omega^2&\omega^2 & \omega & 1\\
    0 & 0 & 0&0 & 1 & 0 & 0 & \omega & \omega^2&\omega^2 & \omega^2 & 1\\
    0 & 0 & 0&0 & 0 & 1 & 0 & 0 & \omega&\omega^2 & \omega & 1
    \end{array}
    \end{pmatrix}.$$
Thus $a=1,b=\omega^2,n=6,t=2,r=0$. According to (1) of Theorem \ref{H-even-0}. Let
$$S=\{ f_1(a-b^{\frac{(q+1)q^2}{2}}(\theta^i+\theta^{-i})),\  1\leq i\leq \frac{n}{2}\}.$$
Computing by Magma \cite{magma}, we obtain that
$$S=\{\xi^{26},\xi^{38},\xi^{41}\},$$
and $C$ has parameters $[12,6,5]$.
Hence $1\notin S$. Therefore, from (1) of Theorem \ref{H-even-0}, $C$ is a Hermitian LCD code, and it is the best-known (see \cite{H-LCD}).
\end{example}

\begin{theorem}\label{H-odd-0}
Let $a\in\F_q$ and $b\in \F_{q^2}$ with $b\neq 0$ and $n\geq 2$ be an integer. Let $C$ be the linear code over $\F_{q^2}$ with the generator matrix $(I_n\ f_1(T_n(a,b,b^q))\ \cdots\ f_{t-1}(T_n(a,b,b^q)))$, where $f_j(x)\in \F_{q^2}[x]\ (1\leq j\leq t-1)$. Assume that ${\rm char}\F_{q^2}=p$ is odd. Let $r$ be the largest integer such that $n+1=p^r(m+1)$ and $b'=b^{\frac{q+1}{2}}$. Then we get the following results.
\begin{itemize}
  \item [(1)] If $r=0$, then
  $C$ is Hermitian LCD if and only if $$-1\notin \{ f_1(a-b'(\theta^i+\theta^{-i}))^2+ \cdots+f_{t-1}(a-b'(\theta^i+\theta^{-i}))^2,\  1\leq i\leq n\},$$
  where $\theta$ is a primitive $2(n+1)$-th root of $1$.
  In particular, when $t=2$, $C$ is Hermitian LCD if and only if $$\mu \notin \{ f_1(a-b'(\theta^i+\theta^{-i})):\  1\leq i\leq n\}\cup \{ -f_1(a-b'(\theta^i+\theta^{-i})):\  1\leq i\leq n\},$$
  where $\mu\in \F_{q^2}$ with $\mu^2=-1$ and $\theta$ is a primitive $2(n+1)$-th root of $1$.
  \item [(2)] If $r\geq 1$ and $m\geq 1$, then $C$ is Hermitian LCD if and only if
  $$-1 \notin \left\{f_1(a+ 2b')^2+\cdots+f_{t-1}(a+ 2b')^2, f_1(a- 2b')^2+\cdots+f_{t-1}(a-2b')^2 \right\}$$ $$\cup \left\{f_1(a-b'(\theta^i+\theta^{-i}))^2+
  \cdots+f_{t-1}(a-b'(\theta^i+\theta^{-i}))^2: 1\leq i\leq m\right\},$$
  where $\theta$ is a primitive $2(m+1)$-th root of $1$.
  In particular, when $t=2$, $C$ is Hermitian LCD if and only if
  $$\mu \notin \{f_1(a+ 2b'),-f_1(a+ 2b'),f_1(a-2b'),-f_1(a-2b')\}$$ $$\cup
  \{f_1(a-b'(\theta^i+\theta^{-i})): 1\leq i\leq m\},$$
  where $\mu\in \F_{q^2}$ with $\mu^2=-1$ and $\theta$ is a primitive $2(m+1)$-th root of $1$.
  \item [(3)] If $r\geq 1$ and $m=0$, then $C$ is Hermitian LCD if and only if
  $$-1 \notin \left\{f_1(a+ 2b')^2+\cdots+f_{t-1}(a+ 2b')^2, f_1(a- 2b')^2+\cdots+f_{t-1}(a-2b')^2 \right\}.$$
  In particular, when $t=2$, $C$ is Hermitian LCD if and only if
  $$\mu \notin \{f_1(a+ 2b'),-f_1(a+ 2b'),f_1(a-2b'),-f_1(a-2b')\},$$
  where $\mu\in \F_{q^2}$ with $\mu^2=-1$.
\end{itemize}
\end{theorem}

\begin{proof}
Note that $(b^{q+1})^{1/2}=b^{\frac{q+1}{2}}$ if $q$ is odd.
The proof is similar to the proof of Theorem \ref{H-even-0}, so we omit it.
\end{proof}

\begin{example}
Let $\F^*_{3^2}=\langle \omega \rangle$ and $f_1(x)=\omega^2 x^4+x^3+x^2\in \F_{3^2}[x]$. Let $C$ be the linear code $C$ over $\F_{3^2}$ with the generator matrix $(I_{6},f_1(T_{6}(1,\omega,\omega^3)))$.
Then $C$ has parameters $[12,6,6]$. By (1) of Theorem \ref{H-odd-0},
$C$ is an optimal Hermitian LCD code over $\F_{3^2}$.
\end{example}

\begin{remark}
In Tables 5 and 6,
we collect some optimal, almost optimal or best-known quaternary Hermitian LCD codes with respect to the newest codetable for quaternary Hermitian LCD codes (see \cite{H-LCD}), where $``*"$ denotes optimal quaternary Hermitian LCD codes, $``\diamond"$ denotes the best known quaternary Hermitian LCD codes. It is worth mentioning that the codes of Table 5 are FSD.
\end{remark}

\begin{center}Table 5: Quaternary Hermitian FSD LCD codes
\begin{tabular}{|c|c|c|}
\hline
$A $&$ f(x) $& Parameters \\
\hline
$T_2(1,\omega,\omega^2) $&$ x $&$ [4,2,2]^*$\\
\hline
$T_4(1,\omega,\omega^2) $&$ \omega x^3+x^2 $&$ [8,4,4]^*$\\
\hline
$T_6(1,\omega,\omega^2) $&$ \omega x^5+x^3 $&$ [12,6,5]^\diamond$\\
\hline
$T_8(1,\omega,\omega^2) $&$ x^6 + \omega x^5 + x $&$ [16,8,6]^\diamond$\\
\hline
$T_{10}(1,\omega,\omega^2) $&$  x^8 +\omega x^7 + x^5 + x^4 + x^3 $&$ [20,10,7]^\diamond$\\
\hline
$T_{12}(1,\omega,\omega^2) $&$  x^{11} +\omega^2 x^9 +\omega x^5 + x^4 $&$ [24,12,8]^\diamond$\\
\hline
\end{tabular}
\end{center}

\begin{center} Table 6: Quaternary Hermitian LCD codes
\begin{tabular}{|c|c|c|c|}
\hline
$A $&$ f_1(x) $&$ f_2(x) $& Parameters \\
\hline
$T_2(1,\omega,\omega^2) $&$ \omega x+1 $&$ x+1 $&$ [6,2,4]^*$\\
\hline
$T_3(1,\omega,\omega^2) $&$ \omega x^2+x $&$ x^2+x +\omega $&$ [9,3,6]^*$\\
\hline
$T_4(1,\omega,\omega^2) $&$ \omega x^3+x^2 $&$ x^3+\omega x^2+x $&$ [12,4,7]^*$\\
\hline
$T_5(1,\omega,\omega^2) $&$ \omega x^4+x^3+1 $&$ \omega x^4+\omega^2x^3+\omega x+1 $&$ [15,5,8]^*$\\
\hline
$T_6(1,\omega,\omega^2) $&$\omega x^5+x^4 $&$ \omega x^5+\omega^2 x^4+\omega^2 x^3+x $&$ [18,6,9]^\diamond$\\
\hline
$T_7(1,\omega,\omega^2) $&$ \omega x^6+x^5 $&$ x^5+x^4+\omega x^3+x^2 $&$ [21,7,10]^\diamond$\\
\hline
\end{tabular}
\end{center}

\subsection{Codes with one-dimension Hermitian hull}

\begin{theorem}\label{H-1}
Let $a\in\F_q$ and $b\in \F_{q^2}$ with $b\neq 0$ and $n\geq 2$ be an integer.
Let $C$ be the linear code over $\F_{q^2}$ with the generator matrix $(I_n\ f_1(T_n(a,b,b^q))\ \cdots\ f_{t-1}(T_n(a,b,b^q)))$, where $f_j(x)\in \F_{q^2}[x]\ (1\leq j\leq t-1)$. Assume that $char\F_{q^2}=p$.
Let $r$ be the largest integer such that $n+1=p^r(m+1)$.
Then we have the following results.
\begin{enumerate}
  \item [(1)] Assume that $char \F_{q^2}$ is even. If $r=1$, $f_1(a)+f_2(a)+\cdots+f_{t-1}(a)=1$ and
  $$f_1(a-b^{\frac{(q+1)q^2}{2}}(\theta^i+\theta^{-i}))+
  \cdots+f_{t-1}(a-b^{\frac{(q+1)q^2}{2}}(\theta^i+\theta^{-i}))\neq 1,1\leq i \leq \frac{m}{2},$$
  where $\theta$ is a primitive $(m+1)$-th root of $1$,
  then $C$ is a linear $[tn,n]$ code with one-dimension Hermitian hull.
  \item [(2)] Assume that $char \F_{q^2}$ is odd.
  \begin{itemize}
    \item When $r=0$, if there is a unique $i\ (1\leq i\leq m)$ such that $$f_1(a-b^{\frac{q+1}{2}}(\theta^i+\theta^{-i}))^2+\cdots+
        f_{t-1}(a-b^{\frac{q+1}{2}}(\theta^i+\theta^{-i}))^2=-1,$$
        where $\theta$ is a primitive $2(m+1)$-th root of $1$.
        Then $C$ is a linear $[tn,n]$ code with one-dimension Hermitian hull.
      In particular, when $t=2$, if there is a unique $i\ (1\leq i\leq m)$ such that $f_1(a-b^{\frac{q+1}{2}}(\theta^i+\theta^{-i}))=\mu$ or $-\mu$, where $\theta$ is a primitive $2(m+1)$-th root of $1$ and $\mu\in \F_{q^2}$ with $\mu^2=-1$.
      Then $C$ is a linear $[2n,n]$ code with one-dimension Hermitian hull.

    \item When $p=3$ and $r=1$, if there is a unique $x$ such that $f_1(x)^2+\cdots+f_{t-1}(x)^2=-1$, where $x\in\{a-2b^{\frac{q+1}{2}},a+2b^{\frac{q+1}{2}}\}$, and $$f_1(a-b^{\frac{q+1}{2}}(\theta^i+\theta^{-i}))^2+\cdots+
        f_{t-1}(a-b^{\frac{q+1}{2}}(\theta^i+\theta^{-i}))^2\neq -1,\ 1\leq i \leq m,$$
  where $\theta$ is a primitive $2(m+1)$-th root of $1$. Then $C$ is a linear $[tn,n]$ code with one-dimension Hermitian hull.
  In particular, when $t=2$, if there is a unique $x$ such that $f_1(x)=\mu$ or $-\mu$, where $x\in\{a-2b^{\frac{q+1}{2}},a+2b^{\frac{q+1}{2}}\}$, and
  $$f_1(a-b^{\frac{q+1}{2}}(\theta^i+\theta^{-i}))\neq \mu\ or\ -\mu,\ 1\leq i \leq m,$$
  where $\theta$ is a primitive $2(m+1)$-th root of $1$ and $\mu\in \F_{q^2}$ with $\mu^2=-1$.
   Then $C$ is a linear $[2n,n]$ code with one-dimension Hermitian hull.
  \end{itemize}
\end{enumerate}
\end{theorem}

\begin{proof}
The proof is similar to that of Theorem \ref{E-1}, the main difference is that $(b^{q+1})^{1/2}=(b^{(q+1)q^2})^{1/2}=b^{\frac{(q+1)q^2}{2}}$ if $char \F_{q^2}$ is even,
$(b^{q+1})^{1/2}=b^{\frac{q+1}{2}}$ $char \F_{q^2}$ is odd.
\end{proof}

\begin{example}
Let $\F^*_{2^2}=\langle \omega \rangle$ and $f_1(x)=\omega x^5+\omega x^2+x\in \F_{2^2}[x]$. Let $C$ be the linear code over $\F_{2^2}$ with the generator matrix $(I_{5},f_1(T_{5}(1,\omega,\omega^2)))$. Then $C$ has parameters $[10,5,4]$.  
Since $f_1(1)=1$ and $f_1(1-\omega^2(\omega+\omega^{-1}))\neq 1$,
$C$ is an $[10,5,4]$ code with one-dimension Hermitian hull over $\F_{2^2}$ by (1) of Theorem \ref{H-1}, which is almost optimal with respect to the Datebase \cite{codetables}.
\end{example}

\section{LCD codes with the generator matrix $(I_n\ T_{n}(a,b,b)^k)$}

In this section, we consider a special class of linear codes with the generator matrix $(I_n\ T_{n}(a,b,b)^k)$ which can be effectively characterized for being LCD.

\begin{definition}\label{def-6.1}
For $a,b\in \F_q$ and an integer $n\geq 2$, let $C^{k}_n(a,b)$ be the $\F_q$-linear code of length $2n$ and dimension $n$ whose generator matrix is the $n\times 2n$ matrix given by
$$(I_n\ T^{k}_n(a,b,b)),$$
where $T^k_n(a,b,b)=\underbrace{T_n(a,b,b)T_n(a,b,b)\cdots T_n(a,b,b)}_k$.
\end{definition}

It is easy to see that the code $C^k_n(a,b)$ is FSD from Theorem \ref{theorem-Thm-FSD}.
Then we generalize the Lemma \ref{lemma-eigenvalue of A}.
\begin{lemma}\label{lemma-A^2k}
Let $A$ be an $n\times n$ matrix over $\F_q$. We have the following cases:
\begin{itemize}
  \item [(1)] Assume that {\rm char}$\F_q$ is even and $k=2^{t}k'$ and $\gcd (k',q)=1$. Then $-1$ is an eigenvalue of $A^{2k}$ if and only if one of $\{\mu^j, 1\leq j\leq k'\}$ is the eigenvalue of $A$, where $\mu$ is a primitive $k'$-th root of $1$.
  \item [(2)] Assume that {\rm char}$\F_q=p$ is odd and $k=p^t{k'}$ with $\gcd (k',q)=1$. Then $-1$ is an eigenvalue of $A^{2k}$ if and only if one of $\{\mu^{2j+1}, 1\leq j\leq 2k'\}$ is the eigenvalue of $A$, where $\mu$ is a primitive $4k'$-th root of $1$.
\end{itemize}
\end{lemma}

\begin{proof}
(1) If char$\F_q$ is even, then
$$A^{2k}+I_n=A^{2^{t+1}k'}+I_n=(A^{k'}+I_n)^{2^{t+1}}=
\left(\prod_{j=1}^{k'}(A+\mu^jI_n)\right)^{2^{t+1}}.$$
Which is from $$x^{2k}+1=x^{2^{t+1}k'}+1=(x^{k'}+1)^{2^{t+1}}=
\left(\prod_{j=1}^{k'}(x+\mu^j)\right)^{2^{t+1}},$$
where $\mu$ is a primitive $k'$-th root of $1$. This completes the proof.

(2) If char$\F_q$ is odd, then we have $\gcd(2k',q)=1$, hence
$$A^{2k}+1=A^{2p^tk'}+I_n=(A^{2k'}-(-I_n))^{p^t}=\left(\prod_{j=1}^{2k'}(A- \mu^{2j+1}I_n)\right)^{p^t}.$$
Which is from $$x^{2k}+1=x^{2p^tk'}+1=(x^{2k'}-(-1))^{p^t}=\left(\prod_{j=1}^{2k'}(x-\mu \xi^{j})\right)^{p^t},$$
where $\mu$ is a primitive $2k'$-th root of $-1$, $\xi$ is a primitive $2k'$-th root of $1$.
So $\xi=\mu^2$ and $\mu$ is a primitive $4k'$-th root of $1$. Implying that
$$x^{2k}+1=\left(\prod_{j=1}^{2k'}(x-\mu^{2j+1})\right)^{p^t}.$$
This completes the proof.
\end{proof}

\begin{theorem}\label{Theorem-LCD-3}
Let $a,b\in \F_q$ with $b\neq 0$ and $n\geq 2$ be an integer. Let $C^k_{n}(a,b)$ the linear code defined as in Definition \ref{def-6.1}. Assume that ${\rm char}\F_q$ is even. Let $r$ be the largest integer such that $n+1=2^r(m+1)$. Assume that $k=2^tk'$. Then we get the following results.
\begin{itemize}
  \item [(1)] If $r=0$, then $C^k_{n}(a,b)$ is LCD if and only if
$$a/b \notin \left\{-\mu^j/b+\theta^i+\theta^{-i}: 1\leq i\leq \frac{n}{2}, 1\leq j\leq k' \right\},$$
where $\mu$ is a primitive $k'$-th root of $1$ and $\theta$ is a primitive $(n+1)$-th root of $1$.
  \item [(2)] If $r\geq 1$ and $m\geq 1$, then $C^k_{n}(a,b)$ is LCD if and only if
$$a/b \notin \left\{-\mu^j/b: 1\leq j\leq k' \right\} \cup \left\{-\mu^j/b+\theta^i+\theta^{-i}: 1\leq i\leq \frac{m}{2}, 1\leq j\leq k' \right\},$$
where $\mu$ is a primitive $k'$-th root of $1$ and $\theta$ is a primitive $(m+1)$-th root of $1$.
  \item [(3)] If $r\geq 1$ and $m=0$, then $C^k_{n}(a,b)$ is LCD if and only if
  $$a \notin \left\{\mu^j: 1\leq j\leq k' \right\},$$
  where $\mu$ is a primitive $k'$-th root of $1$.
\end{itemize}
\end{theorem}

\begin{proof}
We only prove (1), others are similar.
It is well-known that $C^k_{n}(a,b)$ is LCD if and only if $I_n+T_n(a,b,b)^{2k}$ is invertible (see \cite{LCD-Massey}). Hence $C^k_{n}(a,b)$ is LCD if and only if $-1$ is not an eigenvalue of $T_n(a,b,b)^{2k}$.
Using Lemma \ref{lemma-A^2k}, we know that $C^k_{n}(a,b)$ is LCD if and only if none of the elements in $\{\mu^j,\ 1\leq j\leq k'\}$ is the eigenvalue of $T_n(a,b,b)$.
    Using Proposition \ref{prop-DD}, we have $\phi(\mu^j)=E_n(a-\mu^j,b^2)\neq 0,\ 1\leq j\leq k'$.
    Finally using Theorem \ref{Theorem-E_n(x)-2}, we have $a/b \notin \left\{-\mu^j/b+\theta^i+\theta^{-i}: 1\leq i\leq \frac{n}{2}, 1\leq j\leq k' \right\}.$
\end{proof}

\begin{remark}
Theorem \ref{Theorem-LCD-3} is a special case of Theorem \ref{E-even-0}, which generalizes Theorem 2.6 and Theorem 2.9 of \cite{LCD-T-matric}.
\end{remark}

\begin{example}
Let $\F_{64}^*=\langle\omega\rangle$. Let $C^{(3)}_6(1,1)$ be the linear code over $\F_{2}$ with the generator matrix
$$(I_6\ T_6(1,1,1)^3)=\begin{pmatrix}
    \begin{array}{cccccccccccc}
    1 & 0 & 0&0 & 0 & 0 & 0 &  1& 1&1 & 0 & 0\\
    0 & 1 & 0&0 & 0 & 0 & 1 & 1 & 0&1 & 1 & 0\\
    0 & 0 & 1&0 & 0 & 0 & 1 & 0 & 1&0 & 1 & 1\\
    0 & 0 & 0&1 & 0 & 0 & 1 & 1 & 0&1 & 0 & 1\\
    0 & 0 & 0&0 & 1 & 0 & 0 & 1 & 1&0 & 1 & 1\\
    0 & 0 & 0&0 & 0 & 1 & 0 & 0 & 1&1 & 1 & 0
    \end{array}
    \end{pmatrix}.
    $$
    Let $\mu=\omega^{21}$ is a primitive $3$-th root of $1$ and $\theta=\omega^9$ is a primitive $7$-th root of $1$.
    Computing by Magma \cite{magma}, we obtain that
    $$ S=\left\{\mu^{j}+\theta^i+\theta^{-i}: 1\leq i\leq 3, 1\leq j\leq 3 \right\}=
    \{\omega^{9},\omega^{18},\omega^{31},\omega^{36},\omega^{47},\omega^{55},
    \omega^{59},\omega^{61},\omega^{62}\},$$
and $C^{(3)}_6(1,1)$ has parameters $[12,6,4]$.
Since $1\notin S$, $C^{(3)}_6(1,1)$ is optimal LCD (see \cite{bound-13}).
\end{example}

The following corollaries are immediate according Theorem \ref{Theorem-LCD-3}.

\begin{cor}
Assume that $b\neq 0$, $\gcd(n+1,q)=1$ and {\rm char}$\F_q$ is even. Let $k=2^tk'$. If $q>\frac{nk'}{2}$, then there exists $a\in \F_q$ such that $C_n^k(a,b)$ is LCD.
\end{cor}

\begin{cor}
For any $b\in \F_q$ with $b\neq 0$ and an integer $n\geq 2$. Assume that ${\rm char} \F_q$ is even. Let $r$ be the largest integer such that $n+1=2^r(m+1)$. Assume that $r\geq 1$ and $k=2^tk'$. If $q>\frac{mk'}{2}+k'$, then there exists $a\in \F_q$ such that $C_n^k(a,b)$ is LCD.
\end{cor}

\begin{theorem}\label{Theorem-4}
Let $a,b\in \F_q$ with $b\neq 0$ and $n\geq 2$ be an integer. Let $C^k_{n}(a,b)$ be the linear code defined as in \ref{def-6.1}. Assume that ${\rm char}\F_q=p$ is odd. Let $r$ be the largest integer such that $n+1=p^r(m+1)$. Assume that $k=p^tk'$. Then $\gcd(m+1,q)=1$ and we get the following results.
\begin{itemize}
  \item [(1)] If $r=0$, then $C^k_{n}(a,b)$ is LCD if and only if
$$a/b \notin \left\{\mu^{2j+1}/b+\theta^i+\theta^{-i}: 1\leq i\leq n, 1\leq j\leq 2k' \right\},$$
where $\mu$ is a primitive $4k'$-th root of $1$ and $\theta$ is a primitive $2(n+1)$-th root of $1$.
  \item If $r\geq 1$ and $m\geq 1$, then $C^k_{n}(a,b)$ is LCD if and only if
$$a/b \notin \left\{\mu^{2j+1}/b+2: 1\leq j\leq 2k' \right\}\cup
\left\{\mu^{2j+1}/b-2: 1\leq j\leq 2k' \right\}$$
$$ \cup  \left\{\mu^{2j+1}/b+\theta^i+\theta^{-i}: 1\leq i\leq m, 1\leq j\leq 2k' \right\},$$
where $\mu$ is a primitive $4k'$-th root of $1$ and $\theta$ is a primitive $2(m+1)$-th root of $1$.\\
  \item If $r\geq 1$ and $m=0$, then $C^k_{n}(a,b)$ is LCD if and only if
$$a/b \notin \left\{\mu^{2j+1}/b+2: 1\leq j\leq 2k' \right\}\cup
\left\{\mu^{2j+1}/b-2: 1\leq j\leq 2k' \right\},$$
where $\mu$ is a primitive $4k'$-th root of $1$.
\end{itemize}
\end{theorem}

\begin{proof}
We only prove (1), others are similar.
The proof is similar to the proof of Theorem \ref{Theorem-LCD-3}, we only describe the different steps. Using Lemma \ref{lemma-A^2k}, $-1$ is not an eigenvalue of $T_n(a,b,b)^{2k}$ if and only if none of the elements in $\{\mu^{2j+1},\ 1\leq j\leq 2k'\}$ is the eigenvalue of $T_n(a,b,b)$.
    Using Proposition \ref{prop-DD}, we have $\phi(\mu^{2j+1})=E_n(a-\mu^{2j+1},b^2)\neq 0,\ 1\leq j\leq 2k'$.
    Finally using Theorem \ref{Theorem-E_n(x)-2}, we complete the proof.
\end{proof}

\begin{remark}
Theorem \ref{Theorem-4} is a special case of Theorem \ref{E-odd-0}, which generalizes Theorem 2.7 and Theorem 2.10 of \cite{LCD-T-matric}.
\end{remark}

The following corollaries are immediate according to Theorem \ref{Theorem-4}.

\begin{cor}
Assume that $b\neq 0$, $\gcd(n+1,q)=1$ and {\rm char}$\F_q =p$ is odd. Let $k=p^tk'$. If $q>2nk'$, then there exists $a\in \F_q$ such that $C_n^k(a,b)$ is LCD.
\end{cor}

\begin{cor}
For any $b\in \F_q$ with $b\neq 0$ and an integer $n\geq 2$. Assume that ${\rm char} \F_q=p$ is odd. Let $r$ be the largest integer such that $n+1=p^r(m+1)$. Assume that $r\geq 1$ and $k=p^tk'$. If $q>2mk'+4k'$, then there exists $a\in \F_q$ such that $C_n^k(a,b)$ is LCD.
\end{cor}

\begin{example}
Let $C^{(2)}_4(1,1)$ be the linear code over $\F_{5}$ with the generator matrix
$$(I_4\ T_4(1,1,)^2)=\begin{pmatrix}
    \begin{array}{cccccccc}
    1 & 0 & 0 & 0 & 2 & 2 & 1 & 0\\
    0 & 1 & 0 & 0 & 2 & 3 & 2 & 1\\
    0 & 0 & 1 & 0 & 1 & 2 & 3 & 2\\
    0 & 0 & 0 & 1 & 0 & 1 & 2 & 2
    \end{array}
    \end{pmatrix}.
    $$
    Let $\F_{25}^*=\langle\omega\rangle$, $\mu=\omega^{3}$ is a primitive $8$-th root of $1$.
   Computing by Magma \cite{magma}, we obtain that
    $$S=\left\{\mu^{2j+1}/b\pm2,\ 1\leq j\leq 4 \right\}=
    \{\omega,\omega^{4},\omega^{5},\omega^8,\omega^{13},\omega^{16},\omega^{17},\omega^{20}\},$$
and $C^{(2)}_4(1,1)$ has parameters $[8,4,4]$. Since $1\notin S$, $C^{(2)}_4(1,1)$ is optimal LCD by Theorem \ref{Theorem-4}.
\end{example}

\section{Conclusion}

In this paper, we introduced a general method for constructing LCD codes and linear codes with one-dimension hull by generalizing the method in \cite{LCD-T-matric}.
When applying the method to some Toeplitz matrices, we obtained many optimal, almost optimal or the best known binary and ternary Euclidean LCD codes and quaternary Hermitian LCD codes. We also obtained many optimal codes with one-dimension hull. In some cases, these codes are formally self-dual.
Finally, we characterized a special class of codes as LCD codes, which includes \cite{LCD-T-matric} as a special case.

This paper contributes in two folds. One is to provide an systematic construction method of LCD codes and linear codes with one-dimension hull. We can construct optimal codes from a simple matrix.
The other is to improve the previously known lower bound on the largest minimum distance of LCD codes.
It is worth mentioning that we gave three tables about FSD LCD codes.

As future work, one possible extension would be to find some simple matrices, using the method to construct more optimal linear codes with small hulls.


\begin{thebibliography}{1}
\bibitem{A-hull-DM} E. F. Assmus Jr., J. D. Key, Affine and projective planes, Discrete Math., 1990, \textbf {83}(2-3): 161-187.
\bibitem{AH-TLCD} M. Araya, M. Harada, On the classification of linear complementary dual codes, Discrete Math., 2019, {\bf342}(1): 270-278.
\bibitem{AH-BLCD-17-24} M. Araya, M. Harada, On the minimum weights of binary linear complementary dual codes. Cryptogr. Commun., 2020, {\bf 12}(2): 285-300.
\bibitem{T-11-19} M. Araya, M. Harada, K. Saito, On the minimum weights of binary LCD codes and ternary LCD codes, https://arxiv.org/pdf/1908.08661.pdf.
\bibitem{Factoring-DK-Po-} M. Bhargava, M. E. Zieve, Factoring Dickson polynomials over finite fields, Finite Fields Appl., 1999, \textbf{5}(2): 103-111.
    \bibitem{Hermitian-LCD} K. Boonniyoma, S. Jitman, Complementary dual subfield linear codes over finite fields, https://arxiv.org/abs/1605.06827.
\bibitem{phi'} J. Borowska, L. {\L}aci\'nska, J. Rychlewska, On determinant of certain pentadiagonal matrix, Journal of Applied Mathematics and Computational Mechanics, 2013, \textbf{12}(3): 21-26.
\bibitem{magma} W. Bosma, J. Cannon, C. Playoust, The Magma algebra system I: The user language, J. Symbolic Comput., 1997, {\bf 24}: 235-265.
\bibitem{B-LCD-40} S. Bouyuklieva, Optimal binary LCD codes, Des. Codes Cryptogr., 2021, \textbf{89}(11): 2445-2461.

\bibitem{lcd-appl} C. Carlet, S. Guilley, Complementary dual codes for counter-measures
to side-channel attacks, Adv. Math. Commun., 2016, \textbf{10}(1): 131-150.
\bibitem{C-G-O-S-LCD-and-isometry} C. Carlet, C. G$\ddot{{\rm u}}$neri, F. $\ddot{{\rm O}}$zbudak, P. Sol\'e, A new concatenated type construction for LCD codes and isometry codes, Discrete Math., 2018, \textbf{341}(3): 830-835.
\bibitem{one-C-L-M} C. Carlet, C. Li, S. Mesnager, Linear codeswith small hulls in semi-primitive case, Des. Codes Cryptogr., 2019, \textbf{87}(12): 3063-3075.
\bibitem{2-lcd} C. Carlet, S. Mesnager, C. Tang, Y. Qi, New characterization and parametrization of LCD codes, IEEE Trans. Information Theory, 2019, {\bf 65}(1): 39-49.
\bibitem{LCD-equivalent} C. Carlet, S. Mesnager, C. Tang, Y. Qi, R. Pellikaan, Linear codes over $\F_q$ are equivalent to LCD codes for $q > 3$, IEEE Trans. Information Theory, 2018, \textbf {64}(4): 3010-3017.

\bibitem{LCD-5} S. T. Dougherty, J. Kim, B. $\ddot{{\rm O}}$zkaya, L. Sok, P. Sol\'e, The combinatorics of LCD codes: linear programming bound and orthogonal matrices, Int. J. Inf. Coding Theory, 2017, \textbf{4}(2-3): 116-128.
    \bibitem{2-LCD-30} Q. Fu, R. Li, F. Fu, Y. Rao, On the construction of binary optimal LCD codes with short length, Int. J. Found. Comput. Sci., 2019, {\bf30}: 1237-1245.
\bibitem{bound-13} L. Galvez, J. L. Kim, N. Lee, Y. G. Roe, B. S. Won, Some bounds
on binary LCD codes, Cryptogr. Commun., 2018, \textbf{10}(4): 719-728.
\bibitem{codetables} M. Grassl, Bounds on the minimum distance of linear codes and quantum codes, http://www.codetables.de. Accessed 4 Dec 2021.
\bibitem{LCD-6} C. G$\ddot{{\rm u}}$neri, B. $\ddot{{\rm O}}$zkaya, P. Sol\'e, Quasi-cyclic complementary dual codes, Finite Fields Appl., 2016, \textbf {42}: 67-80.
\bibitem{HS-BLCD-1-16} M. Harada, K. Saito, Binary linear complementary dual codes, Cryptogr. Commun., 2019, {\bf11}(4): 677-696.
\bibitem{J-p-group} J. Leon, Permutation group algorithms based on partition I: theory and algorithms, J. Symb. Comput., 1982, \textbf{12}(4-5): 533-583.
\bibitem{Li-one-hull} C. Li, P. Zeng, Constructions of linear codes with one-dimensional hull, IEEE Trans. Information Theory, 2019, {\bf65}(3): 1668-1676.
\bibitem{H-lcd-1} Z. Liu, J. Wang, Further results on Euclidean and Hermitian linear complementary dual codes, Finite Fields Appl., 2019, {\bf 59}: 104-133.
\bibitem{H-LCD} L. Lu, X. Zhan, S. Yang, H. Cao, Optimal quaternary Hermitian LCD codes, https://arxiv.org/pdf/2010.10166.pdf.
\bibitem{LCD-Massey} J. Massey, Linear codes with complementary duals, Discrete Math., 1992, \textbf{106-107}: 337-342.
\bibitem{MacWilliams} F. J. MacWilliams, N. J. A. Sloane, The theory of Error Correcting Codes, Amsterdam. The Netherlands: North-Holland, 1977.

\bibitem{one-qian-ccds} L. Qian, X. Cao, S. Mesnager, Linear codes with one-dimensional hull associated with Gaussian sums, Cryptogr. Commun., 2021, \textbf{13}(2): 225-243.
\bibitem{one-qian} L. Qian, X. Cao, W. Lu, P. Sol\'e, A new method for constructing linear codes with small hulls, Des. Codes Cryptogr., https://doi.org/10.1007/s10623-021-00940-1, (2021)
\bibitem{LCD-qian} L. Qian, M. Shi, P. Sol\'e, On self-dual and LCD quasi-twisted codes of index two over a special chain ring, Cryptogr. Commun., 2019, \textbf{11}(4): 717-734 .
\bibitem{LCD-is-good} N. Sendrier, Linear codes with complementary duals meet the Gilbert-Varshamov bound, Discrete Math., 2004, \textbf{285}(1): 345-347.
\bibitem{Sen-p-e-codes} N. Sendrier, Finding the permutation between equivalent codes: the support splitting algorithm, IEEE Trans. Information Theory, 2000, \textbf{46}(4): 1193-1203.
\bibitem{Sen-S-auto-group} N. Sendrier, G. Skersys, On the computation of the automorphism group of a linear code, in Proc. IEEE Int. Symp. Inf. Theory, Washington, DC, 2001, 13.
\bibitem{LCD-T-matric} M. Shi, F. $\ddot{{\rm O}}$zbudak, L. Xu, P. Sol\'e. LCD codes from tridiagonal Toeplitz matrices, Finite Fields Appl., 2021, {\bf75}(8), 101892.
\bibitem{LCD-9} M. Shi, D. Huang, L. Sok, P. Sol\'e, Double circulant LCD Codes over $\Z_4$, Finite Fields Appl., 2019, \textbf{58}: 133-144.
\bibitem{LCD-huang} M. Shi, D. Huang, L. Sok, P. Sol\'e, Double circulant self-dual and LCD codes over Galois rings, Adv. Math. Commun., 2019, \textbf{13}(1): 171-183.
\bibitem{Li} M. Shi, S. Li, J. Kim, P. Sol\'e, LCD and ACD codes over a noncommutative non-unital ring with four elements, Cryptogr. Commun., https://10.1007/s12095-021-00545-4 (2021).
\bibitem{isodual} M. Shi, L. Xu, P. Sol\'e, Construction of isodual codes from polycirculant matrices, Des. Codes Cryptogr., 2020, \textbf{88}(12): 2547-2560.
\bibitem{double-T-codes} M. Shi, L. Xu, P. Sol\'e, On isodual double Toeplitz codes,
https://arxiv.org/abs/2102.09233.pdf.
\bibitem{zhu} M. Shi, H. Zhu, L. Qian, L. Sok, P. Sol\'e, On self-dual and LCD double circulant and double negacirculant codes over $\F_q+u\F_q$, Cryptogr. Commun., 2020, \textbf{12}: 53-70.
\end{thebibliography}
\end{document}